\documentclass[pra,aps,showpacs,onecolumn,twoside,superscriptaddress]{revtex4}

\usepackage{mathtools}
\usepackage{amsmath,amsfonts,amssymb,caption,color,epsfig,graphics,graphicx,hyperref,latexsym,mathrsfs,revsymb,theorem,url,verbatim,epstopdf,float}
\usepackage{tikz}
\usepackage[level]{datetime}
\usetikzlibrary{shapes, arrows}

\hypersetup{colorlinks,linkcolor={blue},citecolor={red},urlcolor={blue}}

\newtheorem{definition}{Definition}
\newtheorem{proposition}[definition]{Proposition}
\newtheorem{lemma}[definition]{Lemma}

\newtheorem{theorem}[definition]{Theorem}
\newtheorem{corollary}[definition]{Corollary}
\newtheorem{conjecture}[definition]{Conjecture}

\newtheorem{remark}[definition]{Remark}
\newtheorem{example}[definition]{Example}
\newtheorem{question}[definition]{Question}

\def\bcj{\begin{conjecture}}
\def\ecj{\end{conjecture}}
\def\bcr{\begin{corollary}}
\def\ecr{\end{corollary}}
\def\bd{\begin{definition}}
\def\ed{\end{definition}}
\def\bea{\begin{eqnarray}}
\def\eea{\end{eqnarray}}
\def\bem{\begin{enumerate}}
\def\eem{\end{enumerate}}
\def\bex{\begin{example}}
\def\eex{\end{example}}
\def\bim{\begin{itemize}}
\def\eim{\end{itemize}}
\def\bl{\begin{lemma}}
\def\el{\end{lemma}}
\def\bma{\begin{bmatrix}}
\def\ema{\end{bmatrix}}
\def\bpf{\begin{proof}}
\def\epf{\end{proof}}
\def\bpp{\begin{proposition}}
\def\epp{\end{proposition}}
\def\bqu{\begin{question}}
\def\equ{\end{question}}
\def\br{\begin{remark}}
\def\er{\end{remark}}
\def\bt{\begin{theorem}}
\def\et{\end{theorem}}

\def\squareforqed{\hbox{\rlap{$\sqcap$}$\sqcup$}}
\def\qed{\ifmmode\squareforqed\else{\unskip\nobreak\hfil
\penalty50\hskip1em\null\nobreak\hfil\squareforqed
\parfillskip=0pt\finalhyphendemerits=0\endgraf}\fi}
\def\endenv{\ifmmode\;\else{\unskip\nobreak\hfil
\penalty50\hskip1em\null\nobreak\hfil\;
\parfillskip=0pt\finalhyphendemerits=0\endgraf}\fi}
\newenvironment{proof}{\noindent \textbf{{Proof.~} }}{\qed}
\def\Dbar{\leavevmode\lower.6ex\hbox to 0pt
{\hskip-.23ex\accent"16\hss}D}
\makeatletter
\def\url@leostyle{%
  \@ifundefined{selectfont}{\def\UrlFont{\sf}}{\def\UrlFont{\small\ttfamily}}}
\makeatother
\def\bcj{\begin{conjecture}}
\def\ecj{\end{conjecture}}
\def\bcr{\begin{corollary}}
\def\ecr{\end{corollary}}
\def\bd{\begin{definition}}
\def\ed{\end{definition}}
\def\bea{\begin{eqnarray}}
\def\eea{\end{eqnarray}}
\def\bem{\begin{enumerate}}
\def\eem{\end{enumerate}}
\def\bex{\begin{example}}
\def\eex{\end{example}}
\def\bim{\begin{itemize}}
\def\eim{\end{itemize}}
\def\bl{\begin{lemma}}
\def\el{\end{lemma}}
\def\bpf{\begin{proof}}
\def\epf{\end{proof}}
\def\bpp{\begin{proposition}}
\def\epp{\end{proposition}}
\def\bqu{\begin{question}}
\def\equ{\end{question}}
\def\br{\begin{remark}}
\def\er{\end{remark}}
\def\bt{\begin{theorem}}
\def\et{\end{theorem}}

\def\btb{\begin{tabular}}
\def\etb{\end{tabular}}

\newcommand{\nc}{\newcommand}

\def\b{\beta}
\def\g{\gamma}

 \nc{\bbA}{\mathbb{A}} \nc{\bbB}{\mathbb{B}} \nc{\bbC}{\mathbb{C}}
 \nc{\bbD}{\mathbb{D}} \nc{\bbE}{\mathbb{E}} \nc{\bbF}{\mathbb{F}}
 \nc{\bbG}{\mathbb{G}} \nc{\bbH}{\mathbb{H}} \nc{\bbI}{\mathbb{I}}
 \nc{\bbJ}{\mathbb{J}} \nc{\bbK}{\mathbb{K}} \nc{\bbL}{\mathbb{L}}
 \nc{\bbM}{\mathbb{M}} \nc{\bbN}{\mathbb{N}} \nc{\bbO}{\mathbb{O}}
 \nc{\bbP}{\mathbb{P}} \nc{\bbQ}{\mathbb{Q}} \nc{\bbR}{\mathbb{R}}
 \nc{\bbS}{\mathbb{S}} \nc{\bbT}{\mathbb{T}} \nc{\bbU}{\mathbb{U}}
 \nc{\bbV}{\mathbb{V}} \nc{\bbW}{\mathbb{W}} \nc{\bbX}{\mathbb{X}}
 \nc{\bbZ}{\mathbb{Z}}

 \nc{\bA}{{\bf A}} \nc{\bB}{{\bf B}} \nc{\bC}{{\bf C}}
 \nc{\bD}{{\bf D}} \nc{\bE}{{\bf E}} \nc{\bF}{{\bf F}}
 \nc{\bG}{{\bf G}} \nc{\bH}{{\bf H}} \nc{\bI}{{\bf I}}
 \nc{\bJ}{{\bf J}} \nc{\bK}{{\bf K}} \nc{\bL}{{\bf L}}
 \nc{\bM}{{\bf M}} \nc{\bN}{{\bf N}} \nc{\bO}{{\bf O}}
 \nc{\bP}{{\bf P}} \nc{\bQ}{{\bf Q}} \nc{\bR}{{\bf R}}
 \nc{\bS}{{\bf S}} \nc{\bT}{{\bf T}} \nc{\bU}{{\bf U}}
 \nc{\bV}{{\bf V}} \nc{\bW}{{\bf W}} \nc{\bX}{{\bf X}}
 \nc{\bZ}{{\bf Z}}

\nc{\cA}{{\cal A}} \nc{\cB}{{\cal B}} \nc{\cC}{{\cal C}}
\nc{\cD}{{\cal D}} \nc{\cE}{{\cal E}} \nc{\cF}{{\cal F}}
\nc{\cG}{{\cal G}} \nc{\cH}{{\cal H}} \nc{\cI}{{\cal I}}
\nc{\cJ}{{\cal J}} \nc{\cK}{{\cal K}} \nc{\cL}{{\cal L}}
\nc{\cM}{{\cal M}} \nc{\cN}{{\cal N}} \nc{\cO}{{\cal O}}
\nc{\cP}{{\cal P}} \nc{\cQ}{{\cal Q}} \nc{\cR}{{\cal R}}
\nc{\cS}{{\cal S}} \nc{\cT}{{\cal T}} \nc{\cU}{{\cal U}}
\nc{\cV}{{\cal V}} \nc{\cW}{{\cal W}} \nc{\cX}{{\cal X}}
\nc{\cZ}{{\cal Z}}

\nc{\hA}{{\hat{A}}} \nc{\hB}{{\hat{B}}} \nc{\hC}{{\hat{C}}}
\nc{\hD}{{\hat{D}}} \nc{\hE}{{\hat{E}}} \nc{\hF}{{\hat{F}}}
\nc{\hG}{{\hat{G}}} \nc{\hH}{{\hat{H}}} \nc{\hI}{{\hat{I}}}
\nc{\hJ}{{\hat{J}}} \nc{\hK}{{\hat{K}}} \nc{\hL}{{\hat{L}}}
\nc{\hM}{{\hat{M}}} \nc{\hN}{{\hat{N}}} \nc{\hO}{{\hat{O}}}
\nc{\hP}{{\hat{P}}} \nc{\hR}{{\hat{R}}} \nc{\hS}{{\hat{S}}}
\nc{\hT}{{\hat{T}}} \nc{\hU}{{\hat{U}}} \nc{\hV}{{\hat{V}}}
\nc{\hW}{{\hat{W}}} \nc{\hX}{{\hat{X}}} \nc{\hZ}{{\hat{Z}}}

\nc{\hn}{{\hat{n}}}

\def\diag{\mathop{\rm diag}}

\def\max{\mathop{\rm max}}
\def\min{\mathop{\rm min}}

\def\tr{\mathop{\rm Tr}}

\def\ox{\otimes}

\newcommand{\bra}[1]{\langle#1|}
\newcommand{\ket}[1]{|#1\rangle}

\newcommand{\ketbra}[2]{|#1\rangle\!\langle#2|}

\def\Dbar{\leavevmode\lower.6ex\hbox to 0pt
{\hskip-.23ex\accent"16\hss}D}
\usepackage{color}
\begin{document}
\title{Coherence and entanglement in Grover and
	Harrow-Hassidim-Lloyd algorithm}

\newdateformat{ukdate}{\ordinaldate{\THEDAY} \monthname[\THEMONTH] \THEYEAR}
\date{\ukdate\today}


\author{Changchun Feng}
\affiliation{LMIB(Beihang University), Ministry of Education}
\affiliation{ School of Mathematical Sciences, Beihang University, Beijing 100191, China}

\author{Lin Chen}\email[]{linchen@buaa.edu.cn (corresponding author)}
\affiliation{LMIB(Beihang University), Ministry of Education}
\affiliation{ School of Mathematical Sciences, Beihang University, Beijing 100191, China}
\affiliation{International Research Institute for Multidisciplinary Science, Beihang University, Beijing 100191, China}
\author{Li-Jun Zhao}\email[]{zhaolijun@buaa.edu.cn (corresponding author)}
\affiliation{LMIB(Beihang University), Ministry of Education}
\affiliation{ School of Mathematical Sciences, Beihang University, Beijing 100191, China}

\begin{abstract}
Coherence, discord and geometric measure (GM) of entanglement are important tools for measuring physical resources. We compute them at every steps of the Grover's algorithm. We summarize these resources's patterns of change. These resources are getting smaller at the step oracle and are getting bigger or invariant
 at the step diffuser. This result is similar to the entanglement's pattern of change in Grover's algorithm. Furthermore, we compute GM at every steps of the Harrow-Hassidim-Lloyd algorithm.

\end{abstract}

\maketitle

\Large

\section{Introduction}
\label{sec:int}
	Quantum entanglement plays an important role as a physical resource in quantum information processing \cite{1935ES,2000MA,2007Horo}.  It is widely used in various quantum information processing tasks such as quantum computing \cite{2005Experimental}, teleportation \cite{2004Deterministic}, dense coding \cite{2002Quantum}, cryptography \cite{2020Entanglement} and  quantum key distribution \cite{Xu2020}. Quantum coherence constitutes a powerful resource for quantum metrology \cite{2004VG,2014RD} and entanglement creation \cite{2005JK,arXivAS} and is at the root of a number of intriguing phenomena of wide-ranging impact in quantum optics \cite{1963RJ,1991MO,1994AA,1995DF}, quantum information \cite{2000MA}, solid state physics \cite{2012CM}, and thermodynamics \cite{1978LH,2014LA}.
	The representatives of the quantum algorithm are Shor’ factoring \cite{1994PW} and Grover’s search \cite{1996LK} algorithms. A few years ago another  algorithm called Harrow Hassidim-Lloyd (HHL) algorithm was developed. It can compute the inverse of sparse matrix.
	The HHL algorithm is known to be optimal in the matrix inversion task. Grover algorithm is an unstructured search algorithm running on a quantum computer, and is one of the typical algorithms of quantum computing. 
	
	 Quantum entanglement is investigated in  Grover alogorithm or HHL algorithm \cite{2022MR}. In this paper we examine a question: ‘how the coherence, discord and GM change in Grover alogorithm or HHL algorithm?’. In order to explore this issue we firstly concentrate on  the Grover algorithm. We compute the coherence in Subsec. \ref{subsec:coh}. We compute  discord in every steps in Lemmas \ref{le:d_rho1}, \ref{le:d_rho2}, \ref{le:d_rho3} and \ref{le:d_rho4}, respectively. We compute  GM in every steps in Lemmas \ref{le:GM1}, \ref{le:GM2}, \ref{le:GM3} and \ref{le:GM4}, respectively. Then we show the tables of  coherence, discord and GM in Tables \ref{tab:coh}, \ref{tab:dis} and \ref{tab:GM}, respectively. We obtain that the variation trends of these physical quantities are getting smaller in the step oracle and getting bigger or invariant in the step diffuser. Furthermore, we concentrate on the HHL algorithm. We compute the GM in every steps of the HHL algorithm in Lemmas \ref{le:GM1_HHL}, \ref{le:GM2_HHL} and \ref{le:GM3_HHL}, respectively. 
	 
	 In addition, some geometrically motivated entanglement measures have been providing us with new insights into quantum entanglement, e.g. entanglement of formation \cite{PhysRevA.54.3824}, relative entropy of entanglement \cite{PhysRevLett.78.2275,PhysRevA.57.1619}, global robustness \cite{PhysRevA.68.012308,PhysRevA.59.141} and squashed entanglement \cite{2004CM}. Besides providing a simple geometric picture, they are closely related to some operationally motivated entanglement measures, e.g. entanglement of distillation \cite{PhysRevA.54.3824} and entanglement cost \cite{PatrickMHayden_2001}. In the future, we will investigate how these entanglement measures change in quantum algorithms.
	 
	The rest of this paper is organized as follows. In Sec. \ref{sec:pre} we introduce the preliminary facts, such as the definitions about coherence, discord, geometric measure of entanglement and lemmas about geometric measure. In Sec. \ref{sec:gro} we investigate the coherence, discord, geometric measure of the quantum states in the steps of Grover's Algorithm respectively. In Sec. \ref{sec:HHL} we investigate the geometric measure of the quantum states in three steps of HHL Algorithm respectively. Finally  we conclude in Sec. \ref{sec:con}.	
\section{Preliminaries}
\label{sec:pre}
	Quantum coherence is conventionally associated with the capability of a quantum state to exhibit quantum interference phenomena \cite{1995DF}. Frozen coherence is the distance between the quantum state $\rho$ and the incoherent state $\delta_{\rho}$, which is the closest irrelevant state of $\rho$. Then we have
	\begin{eqnarray}
		\mathcal{C}(\rho)=\mathcal{D}(\rho,\delta_{\rho})=\min_{\rho'}\mathcal{D}(\rho,\rho').
	\end{eqnarray} 
		Then we introduce the distance norm. Firstly we introduce the definition of the Frobenius norm \cite{BOTTCHER20081864}, 
	
	\begin{eqnarray}
		||A||_{F}=\sqrt{\sum_{i=1}^{m}\sum_{j=1}^{n}|a_{i,j}|^2}=\sqrt{\tr(A^\dagger A)}.
	\end{eqnarray}
	
	We  use the Frobenius norm as the distance norm as follows,
	\begin{eqnarray}
		\mathcal{D}(A,B)=||A-B||_{F}=\sqrt{\tr ((A-B)^\dagger (A-B))},
	\end{eqnarray}
	  where $A$ and $B$ are two matrices. Frobenius norm is used in the  detection of glottal closure instants \cite{1994MA}.

	Classical quantum discord revolves around information theory \cite{Bera_2018}. If we measure the lack of information by entropy, this definition of correlations is
	captured by the mutual information
	\begin{eqnarray}
	I(A:B)=S(A)+S(B)-S(AB),
	\end{eqnarray}
	 where $S(X)$ is the von Neumann entropy $S(X)=-\tr(\rho_X\log \rho_X)$ and $\rho_X$ is  a quantum state of system $X$. For classical variables, Bayes’ rule defines a conditional probability as $p_{x|y}=p_{xy}/p_y$. This implies an equivalent form for the classical mutual information
	 \begin{eqnarray}
	 	J_{cl}(B|A)=S(B)-S(B|A),
	 \end{eqnarray}	 
 	where the conditional entropy $S(B|A)=\sum_a p_aS(B|a)$ is the
	average of entropies $S(B|a)=-\sum_b p_{b|a}\log p_{b|a}$. The notion of classicality related to quantum discord
	revolves around information theory \cite{2000Zurek,2001Hend,oz01}.
	
	We introduce positive-operator-valued measure (POVM) on subsystem A. The measurement is described by a POVM with elements $E_a=M_a^\dagger M_a$, where $M_a$ is the measurement operator and $a$ is the classical outcome. Moreover we have  $\sum_a E_a=I$. The initial state $\rho_{AB}$ is transformed under the measurement to
	\begin{eqnarray}
		\rho_{AB} \rightarrow \rho_{AB}'=\sum_a (M_a \otimes I_B) \rho_{AB} (M_a \otimes I_B)^\dagger,
	\end{eqnarray}
	where party $A$ observes outcome a with probability 
	\begin{eqnarray}
	\label{eq:dis1}
	p_a=\tr ((E_a \otimes I) \rho_{AB}),
	\end{eqnarray}
	and B has the conditional state 
	\begin{eqnarray}
	\label{eq:dis2}
	\rho_{B|a}=\tr_A((E_a\otimes I) \rho_{AB})/p_a.
	\end{eqnarray}
    Then we define the conditional entropy $S(B|{E_a})=\sum_a p_aS(\rho_{B|a})$. Using Equations \eqref{eq:dis1},\eqref{eq:dis2}  we have
	\begin{equation}
		\label{eq：BEa}
		\begin{aligned}
				J(B|{E_a})&=S(B)-S(B|{E_a})=S(B)-\sum_a p_aS(\rho_{B|a})\\&=S(B)-\sum_a \tr (E_a \rho_{AB}) S({ \tr}_A(E_a \rho_{AB})/\tr (E_a \rho_{AB})). 
	\end{aligned}
	\end{equation}
	
	We will quantify the classical correlations of the state. Then independently of a measurement $J(B|{E_a})$ is maximized over all measurements,
	\begin{eqnarray}
		J(B|A)=\max_{E_a} J(B|{E_a}).
	\end{eqnarray}
	
	The quantum discord of the state $\rho_{AB}$ under the measurement $\{E_a\}$ is defined as a difference between total correlations.
	
	\begin{eqnarray}
		\label{eq:D}
		D(B|A)=I(A:B)-J(B|A)=\min_{\{E_a\}}\sum_a p_a S(\rho_{B|a})+S(A)-S(AB).
	\end{eqnarray}

	GM is closely related to the construction of optimal entanglement witnesses \cite{PhysRevA.68.042307} and discrimination of quantum states under LOCC \cite{PhysRevLett.96.040501,PhysRevA.77.012104,Markham_2007}. In condensed matter physics, GM is useful for studying quantum
	many-body systems, such as characterizing ground state properties and detecting phase
	transitions \cite{odv2008, 2008Orus2}.  We give the definition as follows \cite{2010Additivity}.

	\begin{definition}
		\label{de:GM}
		Suppose $\rho$ is an $N$-qubit state. GM measures the closest distance in terms of overlap between the state $\rho$ and the set of separable states, or, equivalently, the set of pure product states. Formally, GM is defined as 
		\begin{eqnarray}
			\label{eq:GM1111}
			\Lambda^2(\rho):=\max_{\sigma \in SEP} \tr(\rho \sigma)=\max_{\ket{\phi} \in PRO} \bra{\phi}\rho\ket{\phi},
		\end{eqnarray}
		\begin{eqnarray}
			\label{eq:GM2}
			G(\rho)=-2\log\Lambda(\rho).
		\end{eqnarray}
	Here, PRO denotes the set of fully pure product states in the Hilbert space $\otimes^N_{j=1} \mathcal{H}_j $. Any pure
	product state maximizing \eqref{eq:GM1111} is the closest product state of $\rho$.	
	\end{definition}

	Then we introduce two lemmas about GM\cite{2010Additivity}.
	\begin{lemma}
		\label{le:GM_cloest}
		In GM, the closest product state to any $N$-partite pure or mixed symmetric state with $N > 3$ is necessarily symmetric.
	\end{lemma}
	A density matrix is called non-negative if all its entries in the computational basis are non-negative.
	\begin{lemma}
		\label{le:GM_non}
		In GM, the closest product state to a non-negative state can be chosen to be non-negative.
	\end{lemma}
	These two lemmas are useful for finding the closest product state and computing GM in every steps of the Grover algorithm or HHL algorithm.
\section{Grover's Algorithm}
\label{sec:gro}
	Let us consider a set $\mathcal{S}=\{\ket{j}|j=0,1,...,N-1\}$, $\langle j_1\ket{j_2}=\delta_{j_1j_2}$. Grover’s algorithm \cite{1996LK,1997LK} tries to find a particular quantum state $\ket{\psi_G}\in \mathcal{S}$. Grover’s algorithm is made up of three unitary transformations called superposition, oracle and diffuser, respectively. The superposition transforms the initial state $\ket{0}^{\ox n}$ to a superposed state, where all states in $\mathcal{S}$ are superposed with equal probability amplitude. This can be achieved by making use of the Hadamard gate,
	\begin{equation}
		H=\frac{1}{\sqrt{2}}\bma 1 & 1\\
							1 & -1\ema. 
	\end{equation}
	Therefore, after superposition transformation the initial state is changed into
	\begin{equation}
		\ket{s}=H^{\ox n}\ket{0}^{\ox n}= \frac{1}{\sqrt{N}}\sum_{j=0}^{N-1}\ket{j},
	\end{equation}
	where $N=2^n$. The oracle and diffuser are described by the unitary operators $U_{oracle}=I-2\ket{\psi_G} \bra{\psi_G}$ and $U_{diffuser}=2\ket{s}\bra{s}-I$. The oracle changes a sign of $\ket{\psi_G}$ in $\ket{s}$. The diffuser increases the probability amplitude of $\ket{\psi_G}$ from $U_{oracle}\ket{s}$.
	
	Even though Grover’s algorithm is optimal as a quantum searching algorithm \cite{1999CZ}, such
	maximal creation and complete annihilation of entanglement do not occur for large $N$. For example, let us consider the case $N = 8$. If $\psi_{G}=\ket{7}=\ket{111}$, Grover’s algorithm changes the quantum state as

	\begin{eqnarray}
		\label{eq:rho1}
		\ket{\psi_1}=U_{oracle}\ket{s}=\frac{1}{2\sqrt{2}}(\ket{0}+\ket{1}+\ket{2}+\ket{3}+\ket{4}+\ket{5}+\ket{6}-\ket{7}),
	\end{eqnarray}
	\begin{eqnarray}
		\label{eq:rho2}
		\ket{\psi_2}=U_{diffuser}\ket{\psi_1}=\frac{1}{4\sqrt{2}}(\ket{0}+\ket{1}+\ket{2}+\ket{3}+\ket{4}+\ket{5}+\ket{6}+5\ket{7}),
	\end{eqnarray}
	\begin{eqnarray}
		\label{eq:rho3}
		\ket{\psi_3}=U_{oracle}\ket{\psi_2}=\frac{1}{4\sqrt{2}}(\ket{0}+\ket{1}+\ket{2}+\ket{3}+\ket{4}+\ket{5}+\ket{6}-5\ket{7}),
	\end{eqnarray}
	\begin{eqnarray}
		\label{eq:rho4}
		\ket{\psi_4}=U_{diffuser}\ket{\psi_3}=\frac{1}{-8\sqrt{2}}(\ket{0}+\ket{1}+\ket{2}+\ket{3}+\ket{4}+\ket{5}+\ket{6}-11\ket{7}).
	\end{eqnarray}

\subsection{Coherence by the Frobenius norm}
\label{subsec:coh}
In this subsection, we will use the Frobenius norm to calculate the coherence.
Suppose $\delta_{\rho_1}$ is the  closest irrelevant state of $\rho_1$. Then we have $\delta_{\psi_1}=\diag(x_1,x_2,x_3,x_4,x_5,x_6,x_7,x_8)$ and $\sum_{i=1}^{8}x_i=1$, $x_i\geq 0$.

\begin{equation}
	\begin{aligned}
		\mathcal{C}_F(\rho_1)&=\mathcal{D}_{F}(\rho_1,\delta_{\rho_1})=||\rho_1-\delta_{\rho_1}||_{F},
	\end{aligned}
\end{equation}
where
\begin{equation}
	\begin{aligned}
		\rho_1-\delta_{\rho_1}&= \frac{1}{8}\bma 
		1-8x_1&1&1&1&1&1&1&-1\\
		1&1-8x_2&1&1&1&1&1&-1\\
		1&1&1-8x_3&1&1&1&1&-1\\
		1&1&1&1-8x_4&1&1&1&-1\\
		1&1&1&1&1-8x_5&1&1&-1\\
		1&1&1&1&1&1-8x_6&1&-1\\
		1&1&1&1&1&1&1-8x_7&-1\\
		-1&-1&-1&-1&-1&-1&-1&1-8x_8\ema .
	\end{aligned}
\end{equation}

Then we obtain that $\mathcal{C}_F(\rho_1)=\min_{x_i}  \sqrt{\frac{1}{64}\sum_{i=1}^{8}(1-8x_i)^2+\frac{7}{8}}$. We have  $\sum_{i=1}^{8}(1-x_i)^2\geq \frac{(8-8\sum_{i=1}^{8}x_i)^2}{8}=0$. The equation is equal when $1-8x_1=1-8x_2=\cdots=1-8x_8=0$. So we get that $\mathcal{C}_F(\rho_1)=\sqrt{\frac{7}{8}}=\frac{\sqrt{14}}{4}$.

Similarly, we have that

\begin{equation}
	\begin{aligned}
		&\mathcal{C}_F(\rho_2)=\mathcal{D}_{F}(\rho_2,\delta_{\rho_2})=||\rho_2-\delta_{\rho_2}||_{F},
	\end{aligned}
\end{equation}
where
\begin{equation}
	\begin{aligned}
&\rho_2-\delta_{\rho_2}\\=& \frac{1}{32}\bma 
1-32x_1&1&1&1&1&1&1&5\\
1&1-32x_2&1&1&1&1&1&5\\
1&1&1-32x_3&1&1&1&1&5\\
1&1&1&1-32x_4&1&1&1&5\\
1&1&1&1&1-32x_5&1&1&5\\
1&1&1&1&1&1-32x_6&1&5\\
1&1&1&1&1&1&1-32x_7&5\\
5&5&5&5&5&5&5&25-32x_8\ema.
\end{aligned}
\end{equation}
$\mathcal{C}_F(\rho_2)=\min_{x_i} \sqrt{\frac{1}{1024}(\sum_{i=1}^{7}(1-32x_i)^2+(25-32x_8)^2)+\frac{49}{128}}=\frac{7\sqrt{2}}{16}$.  When $x_i=\frac{1}{32}$ for $i \in [1,7]$, $x_8=\frac{25}{32}$, the expression gets the minimum value.
	
For $\rho_3$, we have that
\begin{equation}
	\mathcal{C}_F(\rho_3)=\mathcal{D}_{F}(\rho_3,\delta_{\rho_3})=||\rho_3-\delta_{\rho_3}||_{F},
\end{equation}
where
\begin{equation}
	\begin{aligned}
	&\rho_3-\delta_{\rho_3}\\=&	\frac{1}{128}\bma 
		1-32x_1&1&1&1&1&1&1&-5\\
		1&1-32x_2&1&1&1&1&1&-5\\
		1&1&1-32x_3&1&1&1&1&-5\\
		1&1&1&1-32x_4&1&1&1&-5\\
		1&1&1&1&1-32x_5&1&1&-5\\
		1&1&1&1&1&1-32x_6&1&-5\\
		1&1&1&1&1&1&1-32x_7&-5\\
		-5&-5&-5&-5&-5&-5&-5&25-32x_8\ema.
	\end{aligned}
\end{equation}
$\mathcal{C}_F(\rho_3)=\min_{x_i} \sqrt{\frac{1}{1024}(\sum_{i=1}^{7}(1-32x_i)^2+(25-32x_8)^2)+\frac{49}{128}}=\frac{7\sqrt{2}}{16}$. When $x_i=\frac{1}{32}$ for $i \in [1,7]$, $x_8=\frac{25}{32}$, the expression gets the minimum value.

For $\rho_4$, we have that
\begin{equation}
	\mathcal{C}_F(\rho_4)=\mathcal{D}_{F}(\rho_4,\delta_{\rho_4})=||\rho_4-\delta_{\rho_4}||_{F},
\end{equation}
where
{\large
\begin{equation}
	\begin{aligned}
		&\rho_4-\delta_{\rho_4}\\=&\frac{1}{128}\bma 
		1-128x_1&1&1&1&1&1&1&-11\\
		1&1-128x_2&1&1&1&1&1&-11\\
		1&1&1-128x_3&1&1&1&1&-11\\
		1&1&1&1-128x_4&1&1&1&-11\\
		1&1&1&1&1-128x_5&1&1&-11\\
		1&1&1&1&1&1-128x_6&1&-11\\
		1&1&1&1&1&1&1-128x_7&-11\\
		-11&-11&-11&-11&-11&-11&-11&121-128x_8\ema.
	\end{aligned}
\end{equation}}

$\mathcal{C}_F(\rho_3)=\min_{x_i} \sqrt{\frac{1}{128^2}(\sum_{i=1}^{7}(1-128x_i)^2+(121-128x_8)^2)+\frac{1736}{128^2}}=\frac{\sqrt{434}}{64}$. When $x_i=\frac{1}{128}$ for $i \in [1,7]$, $x_8=\frac{121}{128}$, the expression gets the minimum value.

We make a table to show the change of the coherence in every step in Grover’s algorithm.
\begin{table}[htb]   
	\begin{center}   
		\caption{The coherence of every step in Grover's algorithm}  
		\label{tab:coh} 
		\begin{tabular}{|c|c|c|}    
			\hline   \textbf{Quantum state} & \textbf{Coherence} & \textbf{Value} (keep two demical places) \\   
			\hline   $\ket{\psi_{1}}$ & $\frac{\sqrt{14}}{4}$ &  $0.95$\\
			\hline   $\ket{\psi_{2}}$ & $\frac{7\sqrt{2}}{16}$ & $0.62$\\  
			\hline   $\ket{\psi_{3}}$ & $\frac{7\sqrt{2}}{16}$ & $0.62$\\
			\hline   $\ket{\psi_{4}}$ & $\frac{\sqrt{434}}{64}$ & $0.33$\\ 
			\hline   
		\end{tabular}   
	\end{center}   
\end{table}

In this table, wo conclude that the coherence of state in Grover's algorithm is getting smaller in the step oracle.

We have known that in  Grover’s algorithm, the three-tangle and  concurrences of the  mixed states is also reduced in the step oracle.

\subsection{Discord}

	In this subsection, we will investigate the quantum discord in the process of Grover's Algorithm.
	Firstly, we investigate the discord of the quantum state $\rho_1=\ket{\psi_{1}}\bra{\psi_{1}}$ in \eqref{eq:rho1}. We have the following observation.
	\begin{lemma}
		\label{le:d_rho1}
		
		Suppose $\ket{\psi_1}=U_{oracle}\ket{s}=\frac{1}{2\sqrt{2}}(\ket{0}+\ket{1}+\ket{2}+\ket{3}+\ket{4}+\ket{5}+\ket{6}-\ket{7})$. For $\rho_1=\ket{\psi_1}\bra{\psi_1}$, we denote $A_1,B_1$ and $C_1$ as three subsystems of $\rho_1$. Then the discord
		\begin{eqnarray}
			\label{eq:d_rho1}
			&D(B_1C_1|A_1)=D(A_1C_1|B_1)=D(A_1B_1|C_1)\\
			&=D(C_1|A_1B_1)=D(B_1|A_1C_1)=D(A_1|B_1C_1)\\
			&=-\frac{1}{4}\log{\frac{1}{4}}-\frac{3}{4}\log{\frac{3}{4}}\\&\approx 0.81.
		\end{eqnarray} 
	\end{lemma}
	\begin{proof}
		
	For $\rho_1=\ket{\psi_1}\bra{\psi_1}$. Then we get the following expression by equation \eqref{eq:D} 
	\begin{eqnarray}
		\label{eq:psi1}
		D(B_1C_1|A_1)=\min_{\{E_a\}}\sum_a p_a S(B_1C_1|a)+S(A_1)-S(\rho_1).
	\end{eqnarray}
	\begin{eqnarray}
	\label{eq:psi11}
	D(C_1|A_1B_1)=\min_{\{E_a\}}\sum_a p_a S(C_1|a)+S(A_1B_1)-S(\rho_1).
	\end{eqnarray}
	 We obtain that  $\rho_1=\frac{1}{8}
	 \bma 1&1&1&1&1&1&1&-1\\
	 1&1&1&1&1&1&1&-1\\
	 1&1&1&1&1&1&1&-1\\
	 1&1&1&1&1&1&1&-1\\
	 1&1&1&1&1&1&1&-1\\
	 1&1&1&1&1&1&1&-1\\
	 1&1&1&1&1&1&1&-1\\
	 -1&-1&-1&-1&-1&-1&-1&1\ema\\$,

	  $\rho_{A_1}=\rho_{B_1}=\rho_{C_1}=\frac{1}{8}(3(\ket{0}+\ket{1})(\bra{0}+\bra{1})+(\ket{0}-\ket{1})(\bra{0}-\bra{1}))=\frac{1}{8}\bma 4&2\\
	2&4\ema\\$. 
	
	$\rho_{B_1C_1}=\rho_{A_1B_1}=\rho_{A_1C_1}=\frac{1}{8}((\ket{00}+\ket{01}+\ket{10}-\ket{11})(\bra{00}+\bra{01}+\bra{10}-\bra{11})+(\ket{00}+\ket{01}+\ket{10}+\ket{11})(\bra{00}+\bra{01}+\bra{10}+\bra{11}))=
	\frac{1}{8}\bma 2&2&2&0\\
	2&2&2&0\\
	2&2&2&0\\
	0&0&0&2\ema$.
	
	We get that two eigenvalues of $\rho_{A_1}$ are $\frac{1}{4},\frac{3}{4}$. Four eigenvalues of $\rho_{B_1C_1}$ are $\frac{1}{4},\frac{3}{4},0,0$. So the von Neumann entropy $S(\rho_{A_1})=-\frac{1}{4}\log{\frac{1}{4}}-\frac{3}{4}\log{\frac{3}{4}}=0.81$.
	The von Neumann entropy $S(\rho_{B_1C_1})=-\frac{1}{4}\log{\frac{1}{4}}-\frac{3}{4}\log{\frac{3}{4}}$.
	 
	 For $\min_{\{E_a\}}\sum_a p_a S(B_1C_1|a)$ in \eqref{eq:psi1}, we suppose $\{E_a\}$ is made of $E_{a_1}$, $E_{a_2}$ and $E_{a_3}$. $E_{a_1}$, $E_{a_2}$ and $E_{a_3}$ have the following expression,
	 \begin{equation}
	 E_{a_1}=\diag(1,1,1,1,0,0,0,0),
	 \end{equation}
	 \begin{equation}
	 E_{a_2}=
	 \diag(0,0,0,0,1,1,1,0),
	\end{equation}
	 \begin{equation}
	E_{a_3}=
	 \diag(0,0,0,0,0,0,0,1).
	\end{equation}
	 In this case, we can get the following result by \eqref{eq：BEa}.
	 
	 \begin{eqnarray}	 	
	 	\sum_a p_a S(B_1C_1|a)=\sum_{a_i} \tr (E_{a_i} \rho_1) S(\tr_A(E_{a_i} \rho_1)/\tr (E_{a_i} \rho_1))=0	 	 
	 \end{eqnarray}

	By \eqref{eq:psi1} we get that 	
	\begin{equation}
		\begin{split}
				D(B_1C_1|A_1)&=\min_{\{E_a\}}\sum_a p_aS(B_1C_1|a)+S(A_1)-S(\rho_1)\\
				&=S(A_1)=-\frac{1}{4}\log{\frac{1}{4}}-\frac{3}{4}\log{\frac{3}{4}}\\
				&\approx 0.81.
		\end{split}	
	\end{equation}
	 For $\min_{\{E_a\}}\sum_a p_a S(C_1|a)$ in \eqref{eq:psi1}, we suppose $\{E_a\}$ is made of $E_{a_1'}$ and $E_{a_2'}$. $E_{a_1'}$, $E_{a_2'}$ and $E_{a_3'}$ have the following expression,
	\begin{equation}
		E_{a_1'}=\diag(1,1,1,1,1,1,0,0),
	\end{equation}
	\begin{equation}
		E_{a_2'}=
		\diag(0,0,0,0,0,0,1,0),
	\end{equation}
	\begin{equation}
		E_{a_3'}=
		\diag(0,0,0,0,0,0,0,1).
	\end{equation}
	In this case, we can get the following result by \eqref{eq：BEa}.
	
	\begin{eqnarray}	 	
		\sum_a p_a S(C_1|a)=\sum_{a_i'} \tr (E_{a_i'} \rho_1) S(\tr_{AB}(E_{a_i'} \rho_1)/\tr (E_{a_i'} \rho_1))=0	 	 
	\end{eqnarray}
	
	By \eqref{eq:psi1} we get that 	
	\begin{equation}
		\begin{split}
			D(C_1|A_1B_1)&=\min_{\{E_a\}}\sum_a p_aS(C_1|a)+S(A_1B_1)-S(\rho_1)\\
			&=S(A_1B_1)=-\frac{1}{4}\log{\frac{1}{4}}-\frac{3}{4}\log{\frac{3}{4}}\\
			&\approx 0.81.
		\end{split}	
	\end{equation}
	Then we calculate $D(A_1C_1|B_1)=D(A_1B_1|C_1)=-\frac{1}{4}\log{\frac{1}{4}}-\frac{3}{4}\log{\frac{3}{4}}$ in the same way as calculating $D(A_1B_1|C_1)$, and  calculate $D(B_1|A_1C_1)=D(A_1|B_1C_1)=-\frac{1}{4}\log{\frac{1}{4}}-\frac{3}{4}\log{\frac{3}{4}}$ in the same way as calculating $D(C_1|A_1B_1)$.
	\end{proof}
	Then we consider the discords of three quantum states $\rho_2=\ket{\psi_2}\bra{\psi_2}$, $\rho_3=\ket{\psi_3}\bra{\psi_3}$ and $\rho_4=\ket{\psi_4}\bra{\psi_4}$ in the Grover’s algorithm's unitary transforms.
	\begin{lemma}
		\label{le:d_rho2}		
		Suppose 	$\ket{\psi_2}=U_{diffuser}\ket{\psi_1}=\frac{1}{4\sqrt{2}}(\ket{0}+\ket{1}+\ket{2}+\ket{3}+\ket{4}+\ket{5}+\ket{6}+5\ket{7})$. For $\rho_2=\ket{\psi_2}\bra{\psi_2}$, we denote $A_2,B_2$ and $C_2$ as three subsystems of $\rho_2$. Then the discord 
		\begin{equation}
			\label{eq:d_rho2}
			\begin{aligned}
					&D(B_2C_2|A_2)=D(A_2C_2|B_2)=D(A_2B_2|C_2)\\
				&=D(C_2|A_2B_2)=D(B_2|A_2C_2)=D(A_2|B_2C_2)\\
				&=-\frac{1}{8}(4+\sqrt{13}))\log{\frac{1}{8}(4+\sqrt{13})}-\frac{1}{8}(4-\sqrt{13})\log{\frac{1}{8}(4-\sqrt{13})}\\&\approx 0.28.
			\end{aligned}
		\end{equation}
	\end{lemma}
	\begin{lemma}
	\label{le:d_rho3}		
	Suppose 	$\ket{\psi_3}=U_{oracle}\ket{\psi_2}=\frac{1}{4\sqrt{2}}(\ket{0}+\ket{1}+\ket{2}+\ket{3}+\ket{4}+\ket{5}+\ket{6}-5\ket{7})$. For $\rho_3=\ket{\psi_3}\bra{\psi_3}$, we denote $A_3,B_3$ and $C_3$ as three subsystems of $\rho_3$. Then the discord 
	\begin{equation}
		\label{eq:d_rho3}
		\begin{aligned}
		&D(B_3C_3|A_3)=D(A_3C_3|B_3)=D(A_3B_3|C_3)\\
		&=D(C_3|A_3B_3)=D(B_3|A_3C_3)=D(A_3|B_3C_3)\\
		&=-\frac{1}{16}(8+\sqrt{37}))\log{\frac{1}{16}(8+\sqrt{37})}-\frac{1}{16}(8-\sqrt{37})\log{\frac{1}{16}(8-\sqrt{37})}\\&\approx 0.52		
	\end{aligned}
	\end{equation}
	
	\begin{equation}
		D(B_3C_3|A_3)=S(A_3)=0.52, D(C_3|A_3B_3)=S(A_3B_3)\approx 0.52.
	\end{equation}
	\end{lemma}
	\begin{lemma}
	\label{le:d_rho4}		
	Suppose $\ket{\psi_4}=U_{diffuser}\ket{\psi_3}=\frac{1}{-8\sqrt{2}}(\ket{0}+\ket{1}+\ket{2}+\ket{3}+\ket{4}+\ket{5}+\ket{6}-11\ket{7})$. For $\rho_4=\ket{\psi_4}\bra{\psi_4}$, we denote $A_4,B_4$ and $C_4$ as three subsystems of $\rho_4$. Then the discord 
	\begin{equation}
		\label{eq:d_rho4}
		\begin{aligned}
		&D(B_4C_4|A_4)=D(A_4C_4|B_4)=D(A_1B_1|C_4)\\
		&=D(C_4|A_4B_4)=D(B_4|A_4C_4)=D(A_4|B_4C_4)\\
		&=-\frac{1}{32}(16+\sqrt{229}))\log{\frac{1}{32}(16+\sqrt{229})}-\frac{1}{32}(16-\sqrt{229})\log{\frac{1}{32}(16-\sqrt{229})}\\
		&\approx 0.17
	\end{aligned}
	\end{equation}
	\begin{equation}
D(B_4C_4|A_4)=S(A_4)=0.17, D(C_4|A_4B_4)=S(A_4B_4)\approx 0.17.
	\end{equation}
	\end{lemma}
	The proofs of these three lemmas are similar to proof of Lemma \ref{le:d_rho1}.
Then we make a table to show the change of the discord in every step in Grover’s algorithm.
\begin{table}[htb]   
	\begin{center}   
		\caption{The discord of every step in Grover's algorithm}  
		\label{tab:dis} 
		\begin{tabular}{|c|c|c|}    
			\hline   \textbf{Quantum state} & \textbf{Discord} & \textbf{Value} (keep two demical places) \\   
			\hline   $\ket{\psi_{1}}$ & $-\frac{1}{4}\log{\frac{1}{4}}-\frac{3}{4}\log{\frac{3}{4}}$ &  $0.81$\\
			\hline   $\ket{\psi_{2}}$ & $-\frac{1}{8}(4+\sqrt{13}))\log{\frac{1}{8}(4+\sqrt{13})}-\frac{1}{8}(4-\sqrt{13})\log{\frac{1}{8}(4-\sqrt{13})}$ & $0.28$\\  
			\hline   $\ket{\psi_{3}}$ & $-\frac{1}{16}(8+\sqrt{37}))\log{\frac{1}{16}(8+\sqrt{37})}-\frac{1}{16}(8-\sqrt{37})\log{\frac{1}{16}(8-\sqrt{37})}$ & $0.52$\\
			\hline   $\ket{\psi_{4}}$ & $-\frac{1}{32}(16+\sqrt{229}))\log{\frac{1}{32}(16+\sqrt{229})}-\frac{1}{32}(16-\sqrt{229})\log{\frac{1}{32}(16-\sqrt{229})}$ & $0.17$\\ 
			\hline   
		\end{tabular}   
	\end{center}   
\end{table}

In this table, we conclude that the discord of state in Grover's algorithm is getting smaller in the step oracle. This trend is similar to the the change of the coherence in every step in Grover’s algorithm.

\subsection{Geometric measure of entanglement}
\label{subsec:GM}
	In this subsection, we will investigate GM in the process of Grover's Algorithm.
	Firstly, we investigate GM of the quantum state $\rho_1=\ket{\psi_{1}}\bra{\psi_{1}}$ in \eqref{eq:rho1}. We have the following observation.

	\begin{lemma}
		\label{le:GM1}
	Suppose $\ket{\psi_1}=U_{oracle}\ket{s}=\frac{1}{2\sqrt{2}}(\ket{0}+\ket{1}+\ket{2}+\ket{3}+\ket{4}+\ket{5}+\ket{6}-\ket{7})$. For $\rho_1=\ket{\psi_1}\bra{\psi_1}$, we get that
		$G(\rho_1)\approx 0.56$. 
	\end{lemma}
	\begin{proof}
		We recall the definition of Geometric measure of entanglement in Definition \ref{de:GM}. Then we have 
		\begin{eqnarray}
			\label{eq:GM_proof1}
			\Lambda^2(\rho_1):=\max_{\sigma \in SEP} \tr(\rho_1 \sigma_1)=\max_{\ket{\phi_1} \in PRO} \bra{\phi_1}\rho_1\ket{\phi_1}.
		\end{eqnarray}
	We have known that $\ket{\phi_1}$ is  fully pure product states in the Hilbert space. By Lemma \ref{le:GM_cloest}, we get that $\ket{\phi_1}$  is a symmetric state. Then we suppose  $\ket{\phi_1}=(\cos\alpha\ket{0}+e^{i\beta} \sin\alpha \ket{1})^{\otimes 3}$ is a closest product state. 
	Then we have that
	\begin{equation}
		\label{eq:GM1_p1}
		\begin{aligned}
			\ket{\phi_1}=&\cos^3 \alpha\ket{000}+\cos^2\alpha\sin \alpha e^{i\beta}(\ket{001}+\ket{010}+\ket{100})\\&+ 
			\cos\alpha\sin^2 \alpha e^{2i\beta}(\ket{011}+\ket{101}+\ket{110})+\sin^3 \alpha e^{3i\beta}\ket{111},
		\end{aligned}
	\end{equation}
	\begin{eqnarray}
		\label{eq:GM1_p2}
		\Lambda^2(\rho_1)&=\max_{\ket{\phi_1} \in PRO}\bra{\phi_1}\rho_1\ket{\phi_1}
		=\max_{x_i}\frac{1}{8}|\sum_{i=1}^7x_i-x_8|^2,
	\end{eqnarray}
	where $x_1=\cos^3 \alpha$, $x_2=x_3=x_5=\cos^2\alpha\sin \alpha e^{i\beta}$,
	$x_4=x_6=x_7=\cos\alpha\sin^2 \alpha e^{2i\beta}$, and 
	$x_8=\sin^3 \alpha e^{3i\beta}$. 
	
	Furthermore, we have 
	\begin{equation}
		\label{eq:GM1_p3}	
		\begin{aligned}
		\Lambda^2(\rho_1):=\max_{\ket{\phi_1} \in PRO}\bra{\phi_1}\rho_1\ket{\phi_1}=\max_{a,b}\frac{1}{8}|a+ib|^2=\max_{a,b}\frac{1}{8}(a^2+b^2),
		\end{aligned}
	\end{equation}
	where $a=\cos^3\alpha+3\cos
	^2\alpha\sin \alpha\cos\beta+ 3\cos
	\alpha\sin^2 \alpha\cos2\beta- \sin^3 \alpha\cos3\beta$ and $b=3\cos
	^2\alpha\sin \alpha\sin\beta+ 3\cos
	\alpha\sin^2 \alpha\sin2\beta- \sin^3 \alpha\sin3\beta$. Then we calculate the maximum value of the function  
	\begin{eqnarray}
		f_1=\frac{1}{8}(a^2+b^2).
	\end{eqnarray}
	Using Mathematica we get the maximum value of $f_1$ is approximately $0.67$  , when $\alpha\approx 0.59$ and $\beta\approx 0$.  So the GM of $\rho_1$,
	\begin{eqnarray}
		\label{eq:GM1}
			G(\rho_1)=-2\log\Lambda(\rho_1)\approx 0.56.			
	\end{eqnarray}	
	\end{proof}
	Then we investigate GM of three quantum states $\rho_2=\ket{\psi_2}\bra{\psi_2}$, $\rho_3=\ket{\psi_3}\bra{\psi_3}$ and $\rho_4=\ket{\psi_4}\bra{\psi_4}$ in the Grover’s algorithm. We have the following observations.
	
	\begin{lemma}
		\label{le:GM2}
		Suppose $\ket{\psi_2}=U_{diffuser}\ket{\psi_1}=\frac{1}{4\sqrt{2}}(\ket{0}+\ket{1}+\ket{2}+\ket{3}+\ket{4}+\ket{5}+\ket{6}+5\ket{7})$. For $\rho_2=\ket{\psi_2}\bra{\psi_2}$, we get that
		$G(\rho_2)\approx 0.11$. 
	\end{lemma}
	\begin{proof}
		We recall the definition of Geometric measure of entanglement in Definition \ref{de:GM}. Then we have 
		\begin{eqnarray}
			\label{eq:GM2_proof2}
			\Lambda^2(\rho_2):=\max_{\sigma \in SEP} \tr(\rho_2 \sigma_2)=\max_{\ket{\phi_2} \in PRO} \bra{\phi_2}\rho_2\ket{\phi_2}.
		\end{eqnarray}
		We have known that $\ket{\phi_2}$ is  fully pure product states in the Hilbert space. Using Lemma \ref{le:GM_cloest} and Lemma \ref{le:GM_non},   we get that $\ket{\phi_2}$  is a symmetric and non-negative state. Then we suppose  $\ket{\phi_2}=(\cos\alpha\ket{0}+\sin\alpha \ket{1})^{\otimes 3}$ is a closest product state, where $0\leq \alpha \leq\frac{\pi}{2}$.
		\begin{eqnarray}
			\label{eq:GM2_2}
			\Lambda^2(\rho_2)&=\bra{\phi_2}\rho_2\ket{\phi_2}
			=\frac{1}{32}(\sum_{i=1}^7x_i+5x_8)^2,
		\end{eqnarray}
		where $x_1=\cos^3 \alpha$, $x_2=x_3=x_5=\cos^2\alpha\sin \alpha$,
		$x_4=x_6=x_7=\cos\alpha\sin^2 \alpha$, and 
		$x_8=\sin^3 \alpha$.

		Using Mathematica we get the maximum of $\Lambda^2(\rho_2)$ is $0.92$ when $\alpha\approx 1.28$. So the GM of $\rho_2$,
		\begin{eqnarray}
			\label{eq:GM2}
			G(\rho_2)=-2\log\Lambda(\rho_2)\approx0.11.
		\end{eqnarray} 
	\end{proof}
	
	\begin{lemma}
		\label{le:GM3}
		Suppose $\ket{\psi_3}=U_{oracle}\ket{\psi_2}=\frac{1}{4\sqrt{2}}(\ket{0}+\ket{1}+\ket{2}+\ket{3}+\ket{4}+\ket{5}+\ket{6}-5\ket{7})$. For $\rho_3=\ket{\psi_3}\bra{\psi_3}$, we get that
		$G(\rho_3)\approx 0.24$. 
	\end{lemma}
	\begin{proof}
		We recall the definition of Geometric measure of entanglement in Definition \ref{de:GM}. Then we have 
		\begin{eqnarray}
			\label{eq:GM_proof3}
			\Lambda^2(\rho_3):=\max_{\sigma \in SEP} \tr(\rho_3 \sigma_3)=\max_{\ket{\phi_3} \in PRO} \bra{\phi_3}\rho_3\ket{\phi_3}.
		\end{eqnarray}
		We have known that $\ket{\phi_3}$ is  fully pure product states in the Hilbert space. By Lemma \ref{le:GM_cloest}, we get that $\ket{\phi_1}$  is a symmetric state. Then we suppose  $\ket{\phi_3}=(\cos\alpha\ket{0}+e^{i\beta} \sin\alpha \ket{1})^{\otimes 3}$ is a closest product state. 
		\begin{eqnarray}
			\label{eq:GM3_3}
			\Lambda^2(\rho_3)&=\max_{\ket{\phi_3} \in PRO}\bra{\phi_3}\rho_3\ket{\phi_3}
			=\max_{x_i}\frac{1}{32}|\sum_{i=1}^7x_i-5x_8|^2,
		\end{eqnarray}
		where $x_1=\cos^3 \alpha$, $x_2=x_3=x_5=\cos^2\alpha\sin \alpha e^{i\beta}$,
		$x_4=x_6=x_7=\cos\alpha\sin^2 \alpha e^{2i\beta}$, and 
		$x_8=\sin^3 \alpha e^{3i\beta}$.  
		
		Furthermore, we have 
		\begin{equation}
			\label{eq:GM3_p3}	
			\begin{aligned}
				\Lambda^2(\rho_3)=\max_{\ket{\phi_3} \in PRO}\bra{\phi_3}\rho_3\ket{\phi_3}
				=\max_{a,b}\frac{1}{32}|a+ib|^2
				=\max_{a,b}\frac{1}{32}(a^2+b^2),
			\end{aligned}
		\end{equation}
		where $a=\cos^3\alpha+3\cos
		^2\alpha\sin \alpha\cos\beta+ 3\cos
		\alpha\sin^2 \alpha\cos2\beta-5\sin^3 \alpha\cos3\beta$ and $b=3\cos
		^2\alpha\sin \alpha\sin\beta+ 3\cos
		\alpha\sin^2 \alpha\sin2\beta-5\sin^3 \alpha\sin3\beta$. Then we calculate the maximum value of the function 
		\begin{eqnarray}
			f_3=\frac{1}{32}(a^2+b^2).
		\end{eqnarray}
		Using Mathematica we get the maximum value of $f_3$ is approximately $0.85$  , when $\alpha\approx 1.43$ and $\beta=\pi$.  So the GM of $\rho_3$,			
		\begin{eqnarray}
			\label{eq:GM3}
			G(\rho_3)=-2\log\Lambda(\rho_3)\approx 0.24	.
		\end{eqnarray} 	
	\end{proof}

\begin{lemma}
	\label{le:GM4}
	Suppose $\ket{\psi_4}=U_{diffuser}\ket{\psi_3}=\frac{1}{-8\sqrt{2}}(\ket{0}+\ket{1}+\ket{2}+\ket{3}+\ket{4}+\ket{5}+\ket{6}-11\ket{7})$. For $\rho_4=\ket{\psi_4}\bra{\psi_4}$, we get that
	$G(\rho_4)=0.05$. 
\end{lemma}
\begin{proof}
	We recall the definition of geometric measure of entanglement in Definition \ref{de:GM}. Then we have 
	\begin{eqnarray}
		\label{eq:GM4_proof1}
		\Lambda^2(\rho_4):=\max_{\sigma \in SEP} \tr(\rho_4 \sigma_4)=\max_{\ket{\phi_4} \in PRO} \bra{\phi_4}\rho_4\ket{\phi_4}.
	\end{eqnarray}
	We have known that $\ket{\phi_4}$ is  fully pure product states in the Hilbert space.  By Lemma \ref{le:GM_cloest}, we suppose  $\ket{\phi_4}=(\cos\alpha\ket{0}+e^{i\beta} \sin\alpha \ket{1})^{\otimes 3}$ is a closest product state. 
	\begin{eqnarray}
		\label{eq:GM4_3}
		\Lambda^2(\rho_4)&=\max_{\ket{\phi_4} \in PRO}\bra{\phi_4}\rho_4\ket{\phi_4}
		=\max_{x_i}\frac{1}{128}|\sum_{i=1}^7x_i-11x_8|^2,
	\end{eqnarray}
	where $x_1=\cos^3 \alpha$, $x_2=x_3=x_5=\cos^2\alpha\sin \alpha e^{i\beta}$,
	$x_4=x_6=x_7=\cos\alpha\sin^2 \alpha e^{2i\beta}$, and 
	$x_8=\sin^3 \alpha e^{3i\beta}$.  
	Furthermore,  we have
	\begin{eqnarray}
		\label{eq:GM4_p3}	
		\begin{split}
			\Lambda^2(\rho_4)&=\max_{\ket{\phi_4} \in PRO}\bra{\phi_4}\rho_4\ket{\phi_4}\\
			&=\max_{x_i}\frac{1}{128}(\sum_{i=1}^7x_i-11x_8)^2\\
			&=\max_{a,b}\frac{1}{128}|a+ib|^2\\
			&=\max_{a,b}\frac{1}{128}(a^2+b^2),
		\end{split}
	\end{eqnarray}
	where $a=\cos^3\alpha+3\cos
	^2\alpha\sin \alpha\cos\beta+ 3\cos
	\alpha\sin^2 \alpha\cos2\beta-11\sin^3 \alpha\cos3\beta$ and $b=3\cos
	^2\alpha\sin \alpha\sin\beta+ 3\cos
	\alpha\sin^2 \alpha\sin2\beta-11\sin^3 \alpha\sin3\beta$. Then we calculate the maximum value of the function 
	\begin{eqnarray}
		f_4=\frac{1}{128}(a^2+b^2).
	\end{eqnarray}
	Using Mathematica we get the maximum value of $f_4$ is approximately $0.96$  , when $\alpha\approx 1.64$ and $\beta\approx 0$.  So the GM of $\rho_4$,

	\begin{eqnarray}
		\label{eq:GM4}
		G(\rho_4)=-2\log\Lambda(\rho_4)\approx 0.05.
	\end{eqnarray} 
	\end{proof}
	Then we make a table to show the change of GM in every step in Grover’s algorithm.
	\begin{table}[htb]   
		\begin{center}   
			\caption{The GM of state in every step in Grover's algorithm}  
			\label{tab:GM} 
			\begin{tabular}{|c|c|c|}    
				\hline   \textbf{Quantum state} & \textbf{GM}  \\   
				\hline   $\ket{\psi_{1}}$  &  $0.56$\\
				\hline   $\ket{\psi_{2}}$ & $0.11$\\  
				\hline   $\ket{\psi_{3}}$ & $0.24$\\
				\hline   $\ket{\psi_{4}}$  & $0.05$\\ 
				\hline   
			\end{tabular}   
		\end{center}   
	\end{table}
	In this table, wo conclude that GM of state of the step oracle in Grover's algorithm is getting bigger and that of the step diffuser in Grover's algorithm is getting smaller. The trend of change of GM is similar to discord.
\section{HHL Algorithm}

\label{sec:HHL}
	The HHL algorithm \cite{2009AW} consists of three steps, which are  quantum phase estimation (QPE), R($\lambda^{-1}$)-rotation,
and inverse QPE. These three steps were experimentally and explicitly realized  by selecting a	
linear equation $Ax = b$, where
\begin{eqnarray}
	A=\frac{1}{2}\bma 3 &1\\1&3\ema,  
\end{eqnarray}
\begin{eqnarray}
	b=\bma b_0\\b_1\ema,
\end{eqnarray}
where $b_0^2+b_1^2=1$.

We discuss how efficiently the HHL algorithm utilizes the entanglement as we discussed previously in the Grover’s algorithm. By this reason we will compute
the entanglement at each stage of the HHL algorithm.
The first three-qubit state after the QPE stage is simply
\begin{eqnarray}
	\ket{\psi_1}=\frac{1}{2}[(b_0-b_1)\ket{01}\otimes(\ket{0}-\ket{1})+(b_0+b_1)\ket{10}\otimes(\ket{0}+\ket{1})].
\end{eqnarray}
Then we have
\begin{eqnarray}
	\rho_1=\ketbra{\psi_1}{\psi_1}=\frac{1}{4}\bma
	0&0&0&0&0&0&0&0\\
	0&0&0&0&0&0&0&0\\
	0&0&(b_0-b_1)^2&-(b_0-b_1)^2&b_0^2-b_1^2&b_0^2-b_1^2&0&0\\
	0&0&-(b_0-b_1)^2&(b_0-b_1)^2&b_1^2-b_0^2&b_1^2-b_0^2&0&0\\
	0&0&b_0^2-b_1^2&b_1^2-b_0^2&(b_0+b_1)^2&(b_0+b_1)^2&0&0\\
	0&0&b_0^2-b_1^2&b_1^2-b_0^2&(b_0+b_1)^2&(b_0+b_1)^2&0&0\\
	0&0&0&0&0&0&0&0\\
	0&0&0&0&0&0&0&0\ema	.
\end{eqnarray}
The second three-qubit state's spectral decomposition is 
\begin{eqnarray}
	\label{eq:rho22}
	\rho_2=p\ketbra{\phi_1}{\phi_1}+(1-p)\ketbra{\phi_2}{\phi_2},
\end{eqnarray}
where
\begin{eqnarray}
	\label{eq:rho_2}
	\ket{\phi_1}=\frac{x_1}{\sqrt{2}}(\ket{010}-\ket{011})+\frac{x_2}{\sqrt{2}}(\ket{100}+\ket{101}),\\
	\ket{\phi_2}=-\frac{x_2}{\sqrt{2}}(\ket{010}-\ket{011})+\frac{x_1}{\sqrt{2}}(\ket{100}+\ket{101}),
\end{eqnarray}
\begin{eqnarray}
	p=\frac{1}{2}[1+\sqrt{1-4\beta_1^2\beta_2^2(1-\gamma^2)}], \gamma=\sqrt{(1-C^2)(1-\frac{C^2}{4})}+\frac{C^2}{2},
\end{eqnarray}
\begin{eqnarray}
C=(\sin\frac{\pi}{4}+2\sin\frac{\pi}{8})\approx0.736,
\end{eqnarray}
\begin{eqnarray}
	\beta_1=\frac{1}{\sqrt{2}}(b_0-b_1),\beta_2=\frac{1}{\sqrt{2}}(b_0+b_1),
\end{eqnarray}
and we have
\begin{eqnarray}
	\label{eq:rho2_x}
	x_1=\frac{a_1}{\sqrt{a_1^2+a_2^2}},\qquad x_2=\frac{a_2}{\sqrt{a_1^2+a_2^2}},
\end{eqnarray}
with
\begin{eqnarray}
	a_1=\beta_1[1+\sqrt{1-4\beta_1^2\beta_2^2(1-\g^2)}-2\beta_2^2(1-\g^2)],
\end{eqnarray}
\begin{eqnarray}
	a_2=\beta_2\g[1+\sqrt{1-4\beta_1^2\beta_2^2(1-\g^2)}].
\end{eqnarray}

Then we have
\begin{eqnarray}
	{\large\rho_2
	=\frac{1}{2}
	\bma
	0&0&0&0&0&0&0&0\\
	0&0&0&0&0&0&0&0\\
	0&0&px_1^2+(1-p)x_2^2&-px_1^2-(1-p)x_2^2&(2p-1)x_1x_2&(2p-1)x_1x_2&0&0\\
	0&0&-px_1^2-(1-p)x_2^2&px_1^2+(1-p)x_2^2&(1-2p)x_1x_2&(1-2p)x_1x_2&0&0\\
	0&0&(2p-1)x_1x_2&(1-2p)x_1x_2&(1-p)x_1^2+px_2^2&(1-p)x_1^2+px_2^2&0&0\\
	0&0&(2p-1)x_1x_2&(1-2p)x_1x_2&(1-p)x_1^2+px_2^2&(1-p)x_1^2+px_2^2&0&0\\
	0&0&0&0&0&0&0&0\\
	0&0&0&0&0&0&0&0\ema.	}
\end{eqnarray}

The third three-qubit state's spectral decomposition is 
\begin{eqnarray}
	\label{eq:rho3_y}
	\rho_3=q\ketbra{\varphi_1}{\varphi_1}+(1-q)\ketbra{\varphi_2}{\varphi_2},
\end{eqnarray}
where
\begin{eqnarray}
	q=\frac{1}{2}[1+\sqrt{1-4(AC_2-BC_1)^2}],
\end{eqnarray}
with
\begin{eqnarray}
	A=\frac{1}{2}[(b_0-b_1)\sqrt{1-C^2}+(b_0+b_1)\sqrt{1-\frac{C^2}{4}}],\\
	B=\frac{1}{2}[-(b_0-b_1)\sqrt{1-C^2}+(b_0+b_1)\sqrt{1-\frac{C^2}{4}}],\\
	C_1=C\frac{3b_0-b_1}{4}\qquad C_2=C\frac{-b_0+3b_1}{4}.	
\end{eqnarray}
One can show $A^2 + B^2 + C_1^2 + C_2^2 = b_0^2 + b_1^2 = 1$ explicitly. $\ket{\varphi_1}$ and $\ket{\varphi_2}$ are
\begin{eqnarray}
	\ket{\varphi_1}=\ket{00}\otimes(y_1\ket{0}+y_2\ket{1}), \ket{\varphi_2}=\ket{00}\otimes(-y_2\ket{0}+y_1\ket{1}),
\end{eqnarray}
where
\begin{eqnarray}	
	y_1=\frac{f_1}{\sqrt{f_1^2+f_2^2}},
\end{eqnarray}
\begin{eqnarray}	
 y_2=\frac{f_2}{\sqrt{f_1^2+f_2^2}},
\end{eqnarray}
with
\begin{eqnarray}
	f_1=A^2-B^2+C_1^2+C_2^2+\sqrt{1-4(AC_2-BC_1)^2},\qquad f_2=2(AB+C_1C_2).
\end{eqnarray}
Then we have
\begin{eqnarray}
	\rho_3=
	\bma
	qy_1^2+(1-q)y_2^2&(2q-1)y_1y_2&0&0&0&0&0&0\\
	(2q-1)y_1y_2&(1-q)y_1^2+qy_2^2&0&0&0&0&0&0\\
	0&0&0&0&0&0&0&0\\
	0&0&0&0&0&0&0&0\\
	0&0&0&0&0&0&0&0\\
	0&0&0&0&0&0&0&0\\
	0&0&0&0&0&0&0&0\\
	0&0&0&0&0&0&0&0\ema	.
\end{eqnarray}
\subsection{Geometric measure of three states in HHL algorithm}
Then we investigate the GM of three states in HHL algorithm.
\subsubsection{First state in HHL algorithm}
Firstly, we investigate the first state $\rho_1=\ketbra{\psi_1}{\psi_1}$ in HHL algorithm. 
\begin{lemma}
	\label{le:GM1_HHL}
	Suppose $\ket{\psi_1}=\frac{1}{2}[(b_0-b_1)\ket{01}\otimes(\ket{0}-\ket{1})+(b_0+b_1)\ket{10}\otimes(\ket{0}+\ket{1})]$. For $\rho_1=\ket{\psi_1}\bra{\psi_1}$, we have
	$G(\rho_1)=-2\log(\max \{\frac{|b_0-b_1|^2}{2},\frac{|b_0+b_1|^2}{2}\})$. 
\end{lemma}
\begin{proof}
	We  know that $U^\dag \rho_1 U$ has the same GM with $\rho_1$, where $U$ is a local unitary matrix. Then we choose the $U$ such that
	\begin{eqnarray}
		\ket{\psi_1'}=U\ket{\psi_1}=\frac{1}{\sqrt{2}}[(b_0-b_1)\ket{000}+(b_0+b_1)\ket{111}].
	\end{eqnarray}
	Then we have
	\begin{eqnarray}
		\rho_1'=\ketbra{\psi_1'}{\psi_1'}
		=\frac{1}{2}\bma
		(b_0-b_1)^2&0&0&0&0&0&0&b_0^2-b_1^2\\
		0&0&0&0&0&0&0&0\\
		0&0&0&0&0&0&0&0\\
		0&0&0&0&0&0&0&0\\
		0&0&0&0&0&0&0&0\\
		0&0&0&0&0&0&0&0\\
		0&0&0&0&0&0&0&0\\
		b_0^2-b_1^2&0&0&0&0&0&0&(b_0+b_1)^2\ema.	
	\end{eqnarray}
	We recall  geometric measure of entanglement in Definition \ref{de:GM}. Then we have 
	\begin{eqnarray}
		\label{eq:GM_proof1}
		\Lambda^2(\rho_1'):=\max_{\sigma \in SEP} \tr(\rho_1' \sigma_1)=\max_{\ket{\phi_1} \in PRO} \bra{\phi_1}\rho_1'\ket{\phi_1}=\max \{\frac{|b_0-b_1|^2}{2},\frac{|b_0+b_1|^2}{2}\}.
	\end{eqnarray}
	Then
	\begin{eqnarray}
		\label{eq:GM_proof1}
		G(\rho_1)=G(\rho_1')=-2\log(\max \{\frac{|b_0-b_1|^2}{2},\frac{|b_0+b_1|^2}{2}\}).
	\end{eqnarray}
\end{proof}
\subsubsection{Second state in HHL algorithm}
Then we investigate the second state $\rho_2$ in HHL algorithm.

\begin{lemma}
	\label{le:GM2_HHL}
	We recall  $\rho_2=p\ketbra{\phi_1}{\phi_1}+(1-p)\ketbra{\phi_2}{\phi_2}$ in Eq. \eqref{eq:rho22}, we have
	$G(\rho_2)=-2\log(\max \{f(0),f(x')\})$, where $x'$ is the solution of 
	\begin{equation}
		\begin{aligned}
				f'=&\cos	^5\alpha(a\cos	\alpha-6a\sin\alpha+3b(\sec^2\alpha-1)^{\frac{1}{2}}\sec^2\alpha\tan\alpha\\&-6b(\sec^2\alpha-1)^{\frac{3}{2}}\sin\alpha+6c(\sec^2\alpha-1)^2\sec^2\alpha\tan\alpha-6c(\sec^2\alpha-1)^3\sin\alpha).
		\end{aligned}
		\end{equation}
\end{lemma}
\begin{proof}
	For  $\ket{\phi_1}$ and $\ket{\phi_2}$ in Eq. \eqref{eq:rho_2}. Let unitary matrices $U_1=I\otimes I\otimes V_1$ and $U_2=I\otimes I\otimes V_2$ such that 
	\begin{eqnarray}
		\ket{\phi_1'}=U_1\ket{\phi_1}=\frac{x_1}{\sqrt{2}}\ket{000}+\frac{x_2}{\sqrt{2}}\ket{111},
	\end{eqnarray}
	\begin{eqnarray}
		\ket{\phi_2'}=U_2\ket{\phi_2}=-\frac{x_2}{\sqrt{2}}\ket{000}+\frac{x_1}{\sqrt{2}}\ket{111}.
	\end{eqnarray}
Then we have
\begin{eqnarray}
	\rho_2'
	=\frac{1}{2}
	\bma
	px_1^2+(1-p)x_2^2&0&0&0&0&0&0&(2p-1)x_1x_2\\
	0&0&0&0&0&0&0&0\\
	0&0&0&0&0&0&0&0\\
	0&0&0&0&0&0&0&0\\
	0&0&0&0&0&0&0&0\\
	0&0&0&0&0&0&0&0\\
	0&0&0&0&0&0&0&0\\
	(2p-1)x_1x_2&0&0&0&0&0&0&(1-p)x_1^2+px_2^2\ema.	
\end{eqnarray}
We recall  geometric measure of entanglement in Definition \ref{de:GM}. Then we have 
\begin{eqnarray}
	\label{eq:GM_proof2}
	\Lambda^2(\rho_2)=	\Lambda^2(\rho_2'):=\max_{\sigma \in SEP} \tr(\rho_2' \sigma_2)=\max_{\ket{\phi_2} \in PRO} \bra{\phi_2}\rho_2'\ket{\phi_2}.
\end{eqnarray}
	
	We know that $\ket{\psi_2}$ is a  fully pure product state in the Hilbert space. By Lemma \ref{le:GM_cloest}, we obtain that $\ket{\psi_2}$  is a symmetric state. Suppose  $\ket{\psi_2}=(\cos\alpha\ket{0}+e^{i\beta} \sin\alpha \ket{1})^{\otimes 3}$ is a closest product state.
Then we have that
\begin{equation}\label{eq:GM1_p1}
	\begin{aligned}
		\ket{\psi_2}=&\cos^3 \alpha\ket{000}+\cos^2\alpha\sin \alpha e^{i\beta}(\ket{001}+\ket{010}+\ket{100})\\&+ 
		\cos\alpha\sin^2 \alpha e^{2i\beta}(\ket{011}+\ket{101}+\ket{110})+\sin^3 \alpha e^{3i\beta}\ket{111},
	\end{aligned}
\end{equation}
\begin{equation}
	\label{eq:GM1_p2}
	\begin{aligned}
		\Lambda^2(\rho_2')=&\max_{\ket{\psi_2} \in PRO}\bra{\psi_2}\rho_2'\ket{\psi_2}\\
		=&\max\{\frac{1}{2}((px_1^2+(1-p)x_2^2)\cos^6\alpha+((1-p)x_1^2+px_2^2)\sin^6\alpha \\&+(4p-2)x_1x_2\cos^3\alpha\sin^3\alpha (e^{3i\beta}+e^{-3i\beta})) \}
		\\ \leq&\max\{\frac{1}{2}((px_1^2+(1-p)x_2^2)\cos^6\alpha+((1-p)x_1^2+px_2^2)\sin^6\alpha \\&+\lvert(8p-4)x_1x_2\cos^3\alpha\sin^3\alpha \rvert \},
	\end{aligned}
\end{equation}
Then we will transform this problem to find the maximum value of $f$, i.e.,
\begin{eqnarray}
	f=a\cos^6\alpha +b\cos^3\alpha\sin^3\alpha +c\sin^6\alpha,
\end{eqnarray}
where
\begin{eqnarray}
	a=\frac{1}{2}(px_1^2+(1-p)x_2^2),
\end{eqnarray}
\begin{eqnarray}
	b=\lvert(4p-2)x_1x_2 \rvert,
\end{eqnarray}
\begin{eqnarray}
	c=\frac{1}{2}((1-p)x_1^2+px_2^2).
\end{eqnarray}
Then we have 
\begin{eqnarray}
	f=\cos^6\alpha(a+b(\sec^2\alpha-1)^{\frac{3}{2}}+c(\sec^2\alpha-1)^{3}).
\end{eqnarray}
Then
\begin{equation}
\begin{aligned}	
	f'=&\cos	^5\alpha(a\cos	\alpha-6a\sin\alpha+3b(\sec^2\alpha-1)^{\frac{1}{2}}\sec^2\alpha\tan\alpha\\&-6b(\sec^2\alpha-1)^{\frac{3}{2}}\sin\alpha+6c(\sec^2\alpha-1)^2\sec^2\alpha\tan\alpha-6c(\sec^2\alpha-1)^3\sin\alpha).
\end{aligned}	
\end{equation}
Suppose $x'$ is the solution of $f'=0$.
Then we have
\begin{eqnarray}
	\label{eq:GM_proof2}
	G(\rho_2)=G(\rho_2')=-2\log(\max \{f(0),f(x')\}).
\end{eqnarray}
\end{proof}

\subsubsection{Third state in HHL algorithm}
Then we investigate the third state $\rho_3$ in HHL algorithm.
 \begin{lemma}
 	\label{le:GM3_HHL}
 	For $\rho_3=q\ketbra{\varphi_1}{\varphi_1}+(1-q)\ketbra{\varphi_2}{\varphi_2}$  in Eq. \eqref{eq:rho3_y}, we obtain that
 	$G(\rho_3)=-2log(\max \{q,1-q\})$.
 \end{lemma}
\begin{proof}
We know that $\rho_3=q\ketbra{\varphi_1}{\varphi_1}+(1-q)\ketbra{\varphi_2}{\varphi_2}$.
We recall geometric measure of entanglement in Definition \ref{de:GM}. Then we have 
\begin{eqnarray}
	\label{eq:GM_proof3}
	\Lambda^2(\rho_3):=\max_{\sigma \in SEP} \tr(\rho_3 \sigma_2)=
	\max \{q,1-q\}.
\end{eqnarray}
Then
\begin{eqnarray}
	\label{eq:GM_proof3}
	G(\rho_3)=-2\log(\max \{q,1-q\}).
\end{eqnarray}
\end{proof}
\section{conclusion}
\label{sec:con}
We have investigated the coherence, discord and GM of quantum states in the steps of Grover's  Algorithm respectively. Then we show the tables of  coherence, discord and GM in Tables \ref{tab:coh}, \ref{tab:dis} and \ref{tab:GM} respectively. We also conclude that the variation trends of these physical quantitie are getting smaller in the step oracle and are getting bigger or invariant at the step diffuser. Then we investigate GM of quantum states in three steps of HHL Algorithm in Lemmas \ref{le:GM1_HHL}, \ref{le:GM2_HHL} and \ref{le:GM3_HHL}, respectively.

These results help investigate the Grover's Algorithm and HHL Algorithm. In the future we plan to investigate more physical quantities about quantum states in Grover's  Algorithm and HHL Algorithm.
\section*{Acknowledgements}

LC was supported by the NNSF of China (Grant No. 11871089). LJZ was supported by the  NNSF of China (Grant No. 12101031), and the Fundamental Research Funds for the Central Universities (Grant Nos. KG12080401 and ZG216S1902).
\bibliographystyle{unsrt}
\bibliography{changchun}
\end{document}